\documentclass[draftclsnofoot,onecolumn]{IEEEtran}
\usepackage{amsmath,cite,amsfonts,amssymb,psfrag}
\usepackage{graphicx}
\usepackage{rotating}
\usepackage{citesort}

\newtheorem{remark}{Remark}
\newtheorem{definition}{Definition}
\newtheorem{proposition}{Proposition}
\newtheorem{lemma}{Lemma}

\newtheorem{proof}{Proof}

\begin{document}
\title{Lossy Source Coding with Reconstruction Privacy}
\IEEEoverridecommandlockouts
\author{
\authorblockN{Kittipong Kittichokechai, Tobias J. Oechtering, and Mikael Skoglund}\\
\authorblockA{School of Electrical Engineering and the ACCESS Linnaeus Center\\
KTH Royal Institute of Technology, Stockholm, Sweden
}
}
\maketitle
\begin{abstract}
 We consider the problem of lossy source coding with side information under a privacy constraint that the re-construction sequence at a decoder should be kept secret to a certain extent from another terminal such as an eavesdropper, a sender, or a helper. We are interested in how the reconstruction privacy constraint at a particular terminal affects the rate-distortion tradeoff. In this work, we allow the decoder to use a random mapping, and give inner and outer bounds to the rate-distortion-equivocation region for different cases where the side information is available non-causally and causally at the decoder. In the special case where each reconstruction symbol depends only on the source description and current side information symbol, the complete rate-distortion-equivocation region is provided. A binary example illustrating a new tradeoff due to the new privacy constraint, and a gain from the use of a stochastic decoder is given.
\end{abstract}
\section{Introduction}\label{sec:Chap_ENDUSER_introduction}
With the emergence of Internet of Things (IoT), the advance of cloud computing, and the growing predominance of smart devices, we are transitioning into a future scenario where almost everyone and everything will be connected. Significant amount of data will be exchanged among users and service providers which inevitably leads to a privacy concern. A user in the network could receive different versions of certain information from different sources. Apart from being able to process the information efficiently, the user may also wish to protect the privacy of his/her action which is taken based on the received information. In this work, we address the privacy concern of the final action/decision taken at the end-user in an information theoretic setting. More specifically, we consider the problem of lossy source coding under the privacy constraint of the end-user (decoder) whose goal is to reconstruct a sequence subject to a distortion criterion. The privacy concern of the end-user may arise due to the presence of an external eavesdropper or a legitimate terminal such as a sender or a helper who is curious about the final reconstruction. We term the privacy criterion as \emph{end-user privacy}, and use the normalized equivocation of the reconstruction sequence at a particular terminal as a privacy measure.

\begin{figure}[t]
    \centering
    \psfrag{X1}[][]{\small{$X_1^{n}$}}
    \psfrag{X2}[][]{\small{$X_2^{n}$}}
    \psfrag{X3}[][]{\small{$X_3^{n}$}}
    \psfrag{C}[][]{\small{Central unit}}
    \psfrag{r1}[][]{\small{$R_1$}}
    \psfrag{r2}[][]{\small{$R_2$}}
    \psfrag{r3}[][]{\small{$R_3$}}
    \psfrag{h1}[][]{\small{$H(\hat{X}^n|X_1^n)/n \geq \triangle_1$}}
    \psfrag{h2}[][]{\small{$H(\hat{X}^n|X_2^n)/n \geq \triangle_2$}}
    \psfrag{h3}[][]{\small{$H(\hat{X}^n|X_3^n)/n \geq \triangle_3$}}
    \psfrag{Xhat}[][]{\small{$\hat{X}^{n}, Ed(F^{(n)}(X_1^{n},X_2^{n},X_3^{n}),\hat{X}^{n}) \leq D$}}
    \includegraphics[width=0.65\textwidth]{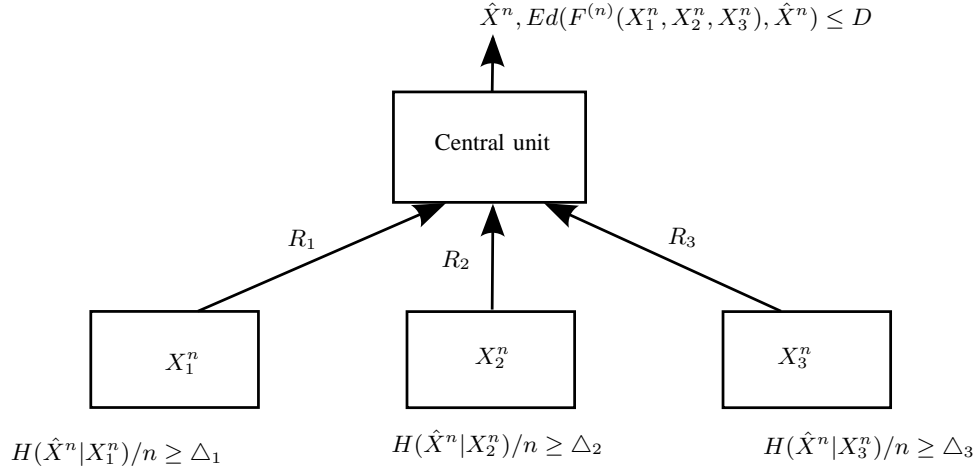}
    \caption{Multiterminal source coding with end-user privacy.}
    \label{fig:Chap_ENDUSER_intro}
\end{figure}

Let us consider Fig. \ref{fig:Chap_ENDUSER_intro} where there exist several agents collecting information for the central unit. Assuming that the agents communicate efficient representations of the correlated sources to the central unit through rate-limited noiseless links so that the central unit is able to estimate a value of some function of the sources $F^{(n)}(X_1^{n},X_2^{n},X_3^{n})$ satisfying the distortion criterion.  However, there is a privacy concern regarding the reconstruction sequence (final decision/action) at the central unit, that it should be kept secret from the agents. This gives rise to a new tradeoff between the achievable rate-distortion pair and privacy of the reconstruction sequence. That is, the central unit should reconstruct a sequence in such a way that it satisfies both distortion and equivocation constraints which can be contradicting. Potential applications of the illustrated setting include those in the area of distributed cloud services where the end-user (central unit) can process information received from the cloud service providers (agents), while guaranteeing that his/her final action will be kept private from the providers, at least to a certain extent.

In this work, we study a special case of Fig. \ref{fig:Chap_ENDUSER_intro} where there are two sources, one of which is available directly at the decoder. For example, we let $X^n$ be the source to be encoded, and $Y^n$ be the uncoded source available at the decoder (see Fig. \ref{fig:Chap_ENDUSER_eve}). Alternatively, we may view $Y^n$ as correlated side information provided by a \emph{helper}\footnote{Here we term a node who only has access to $Y^n$ as a \emph{helper} because it connects to the setting in Fig. \ref{fig:Chap_ENDUSER_intro} in a broader sense.}. The reconstruction sequence $\hat{X}^n$ is an estimate of the value of some component-wise function $F^{(n)}(X^n,Y^n)$, where the $i^{\text{th}}$ component $F^{(n)}_i(X^n,Y^n)=F(X_i,Y_i)$ for $i=1,\ldots,n$. Without the end-user privacy constraint, this corresponds to the problem of source coding with side information at the decoder or the Wyner-Ziv problem \cite{WynerZiv}, \cite{Yamamoto}. We consider three scenarios where the end-user privacy constraint is imposed at different nodes, namely the eavesdropper, the encoder, and the helper, as shown in Fig. \ref{fig:Chap_ENDUSER_eve}, \ref{fig:Chap_ENDUSER_enc}, and
\ref{fig:Chap_ENDUSER_help}. Since the goal of end-user privacy is to protect the reconstruction sequence generated at the decoder against any unwanted inferences, we allow the decoder mapping to be a random mapping. It can be shown by an example that a stochastic decoder can enlarge the rate-distortion-equivocation region as compared to the one derived for deterministic decoders.\footnote{Although the use of a stochastic \emph{encoder} might also help especially if the decoder is deterministic, we restrict ourself to the deterministic encoder here.
Conservatively, it might be reasonable to assume that only the end-user is willing to implement a new coding scheme (stochastic decoder) to improve his/her privacy.
For the case of memoryless reconstruction in Fig. \ref{fig:Chap_ENDUSER_eve}, it can be shown that allowing the use of a stochastic encoder does not improve the rate-distortion-equivocation region.}

\subsection{Overview of Problem Settings and Organization}
We study an implication of the end-user privacy constraint on the rate-distortion tradeoff where the privacy constraint is imposed at different nodes in the system. A summary of contribution is given below.
\begin{itemize}
\item Section \ref{sec:Chap_ENDUSER_eve} considers end-user privacy at the eavesdropper, as depicted in Fig. \ref{fig:Chap_ENDUSER_eve}. It corresponds to a scenario where there is an eavesdropper observing the source description and its side information, and we wish to prevent it from inferring the final reconstruction. We give inner and outer bounds to the rate-distortion-equivocation region for the cases where the side information is available non-causally and causally at the decoder. In a special case of causal side information where the decoder has no \emph{memory}, that is, each reconstruction symbol depends only on the source description and current side information symbol, the complete characterization of the rate-distortion-equivocation region is given. A binary example illustrating the potential gain from allowing the use of a stochastic decoder is also given in Section \ref{sec:Chap_ENDUSER_example}.

 We note that the case of end-user privacy at the encoder in Fig. \ref{fig:Chap_ENDUSER_enc} is included Fig. \ref{fig:Chap_ENDUSER_eve} when $Z^n=X^n$ since the encoder is a deterministic encoder.
 The results can therefore be obtained straightforwardly from those of the setting in Fig. \ref{fig:Chap_ENDUSER_eve}.
\item Section \ref{sec:Chap_ENDUSER_help} considers end-user privacy at the helper, as shown in Fig. \ref{fig:Chap_ENDUSER_help}. It corresponds to a scenario where we wish to prevent the helper from inferring the final reconstruction. Inner and outer bounds to the rate-distortion-equivocation region are given.
\end{itemize}

\begin{figure}[t]
    \centering
    \psfrag{x}[][]{\small{$X^{n}$}}
    \psfrag{y}[][]{\small{$Y^{n}$}}
    \psfrag{z}[][]{\small{$Z^{n}$}}
    \psfrag{w}[][]{\small{$W$}}
    \psfrag{enc}[][]{\small{Encoder}}
    \psfrag{h}[][]{\small{$\frac{1}{n}H(\hat{X}^n|W,Z^n) \geq \triangle$}}
    \psfrag{dec}[][]{\small{Decoder}}
    \psfrag{eve}[][]{\small{Eve}}
    \psfrag{xhat}[][]{\small{$\hat{X}^{n}$}}
    \psfrag{d}[][]{\small{$Ed(F^{(n)}(X^{n},Y^{n}),\hat{X}^{n}) \leq D$}}
    \includegraphics[width=0.65\textwidth]{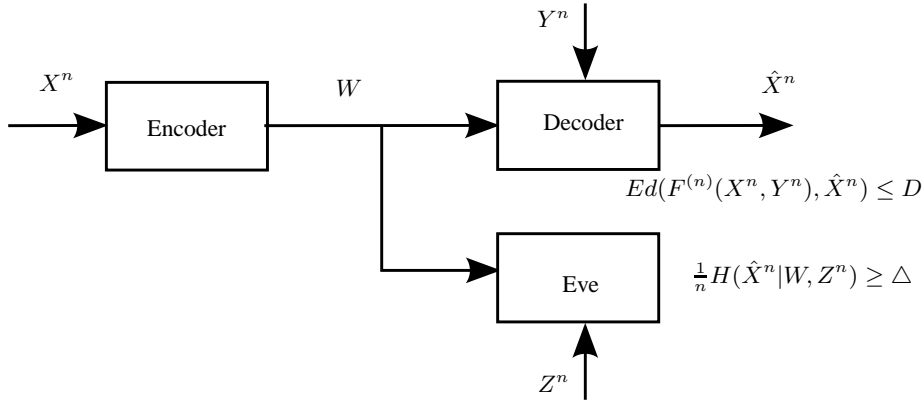}
    \caption{Source coding with end-user privacy at eavesdropper.}\label{fig:Chap_ENDUSER_eve}
\end{figure}
\begin{figure}[t]
    \centering
    \psfrag{x}[][]{\small{$X^{n}$}}
    \psfrag{y}[][]{\small{$Y^{n}$}}
    \psfrag{w}[][]{\small{$W$}}
    \psfrag{enc}[][]{\small{Encoder}}
    \psfrag{h}[][]{\small{$\frac{1}{n}H(\hat{X}^n|X^n) \geq \triangle$}}
    \psfrag{dec}[][]{\small{Decoder}}
    \psfrag{xhat}[][]{\small{$\hat{X}^{n}$}}
    \psfrag{d}[][]{\small{$Ed(F^{(n)}(X^{n},Y^{n}),\hat{X}^{n}) \leq D$}}
    \includegraphics[width=0.65\textwidth]{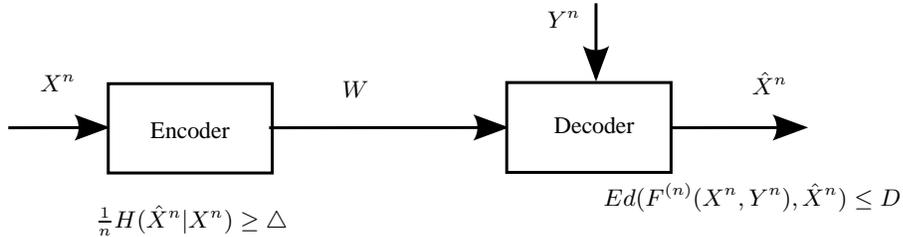}
    \caption{Source coding with end-user privacy at encoder.}\label{fig:Chap_ENDUSER_enc}
\end{figure}
\begin{figure}[t]
    \centering
    \psfrag{x}[][]{\small{$X^{n}$}}
    \psfrag{y}[][]{\small{$Y^{n}$}}
    \psfrag{h}[][]{\small{$\frac{1}{n}H(\hat{X}^n|Y^n) \geq \triangle$}}
    \psfrag{w}[][]{\small{$W$}}
    \psfrag{enc}[][]{\small{Encoder}}
    \psfrag{dec}[][]{\small{Decoder}}
    \psfrag{xhat}[][]{\small{$\hat{X}^{n}$}}
    \psfrag{d}[][]{\small{$Ed(F^{(n)}(X^{n},Y^{n}),\hat{X}^{n}) \leq D$}}
    \includegraphics[width=0.65\textwidth]{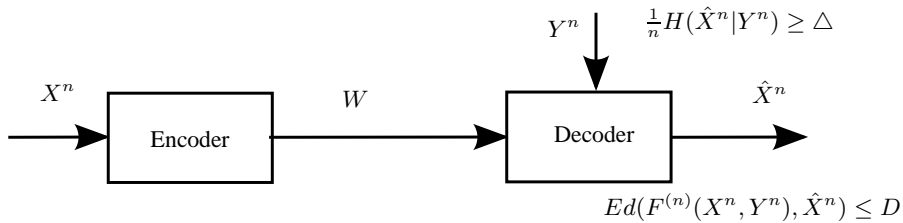}
    \caption{Source coding with end-user privacy at helper.}\label{fig:Chap_ENDUSER_help}
\end{figure}
\subsection{Related Work}
The idea of protecting the reconstruction sequence against an eavesdropper was first considered as an additional secrecy constraint in the context of coding for watermarking and encryption by Merhav in \cite{Merhav_a} where the author considered a watermarking setting using a secret key sequence to protect the (watermark) message and reconstruction sequences. It was also considered in a related Shannon cipher system where the secret key is distributed through a capacity-limited channel in \cite{Merhav_b}. Recently, Schieler and Cuff in \cite{Schieler} considered a lossy source coding setting with common secret key and the objective is to maximize a payoff function based on the source, legitimate's and eavesdropper's reconstruction sequences. Under certain assumptions, the payoff function can reduce to the equivocation of the reconstruction sequence. Although it was discussed in \cite{EkremUlukus} that the end-user privacy constraint which is the equivocation bound of the reconstruction sequence might be an inconsistent measure of the \emph{source secrecy}, in our work, it is still a reasonable measure from an end-user's secrecy point of view as it measures amount of the remaining uncertainty of the reconstruction sequence at a particular terminal.  Closely related to the end-user privacy, Tandon et. al in \cite{TandonSankarPoor} considered the setting of Heegard-Berger lossy source coding \cite{HeegardBerger} where the degraded decoder has an additional privacy constraint on the side information of the stronger decoder.
With the focus on source secrecy, secure lossless distributed source coding was studied by Prabhakaran and Ramchandran \cite{PrabhakaranRamchandran}, G\"{u}nd\"{u}z et al.  \cite{GunduzErkipPoor}, and Tandon et al. \cite{TandonUlukusRamchandran}. Villard and Piantanida in \cite{VillardPiantanida} considered the extension to the lossy setting and characterized the optimal tradeoff between rate, distortion, and equivocation rate of the source for some special cases. Notations used in the paper follow standard ones in \cite{ElGamalKim}.

\section{End-user Privacy at Eavesdropper}\label{sec:Chap_ENDUSER_eve}

\subsection{Problem Formulation}
We consider a setting in the presence of an external eavesdropper, as shown in Fig. \ref{fig:Chap_ENDUSER_eve}. Source, side information, and reconstruction alphabets, $\mathcal{X}, \mathcal{Y}, \mathcal{Z}, \hat{\mathcal{X}}$ are assumed to be finite. Let $(X^{n},Y^n,Z^n)$ be $n$-length sequences which are i.i.d. according to $P_{X,Y,Z}$. A function $F^{(n)}(X^n,Y^n)$ is assumed to be a component-wise function, where the $i^{\text{th}}$ component $F^{(n)}_i(X^n,Y^n)=F(X_i,Y_i)$ with $F:\mathcal{X} \times \mathcal{Y} \rightarrow \mathcal{F}$, for $i=1,\ldots,n$ (cf., e.g., \cite{Yamamoto}). Given a source sequence $X^n$, an encoder generates a source description $W \in \mathcal{W}^{(n)}$ and sends it over the noise-free, rate-limited link to a decoder. Given the source description and the side information $Y^n$, the decoder randomly generates $\hat{X}^n$ as an estimate of the value of the function $F^{(n)}(X^n,Y^n)$ such that it satisfies a distortion criterion. The eavesdropper has access to the source description and its own side information $Z^n$. The end-user privacy at the eavesdropper is then measured by the normalized conditional entropy $H(\hat{X}^n|W,Z^n)/n$. We are interested in characterizing the optimal tradeoff between rate, distortion, and equivocation of the reconstruction sequence in terms of the rate-distortion-equivocation region.

The model in Fig. \ref{fig:Chap_ENDUSER_eve} is similar to the secure source coding with side information in \cite{VillardPiantanida}, except that the end-user privacy is imposed instead of the source privacy. The setting is also closely related to the model of side information privacy studied in \cite{TandonSankarPoor} where the authors are interested in the privacy of side information at the second decoder who is also required to decode the source subject to a distortion constraint.
As for the end-user privacy, \cite{Merhav_a} considered a similar constraint in the context of coding for watermarking and encryption. The main differences to our setting are that the author considered the case where there exists a common secret key sequence independent of the message sequence at both encoder and decoder, and that the use of a stochastic decoder was not considered. From the problem formulation point of view, the end-user privacy constraint can also be considered as a complement to the \emph{common reconstruction constraint} in lossy source coding problems \cite{Steinberg},\cite{Lapidoth} where the reconstruction sequence is instead required to be reproduced at the sender.

Definitions of code, achievability, and the rate-distortion-equivocation region are given below.
\begin{definition}
A $(|\mathcal{W}^{(n)}|,n)$-code for source coding with end-user privacy consists of
\begin{itemize}
  \item an encoder $f^{(n)}: \mathcal{X}^{n} \rightarrow \mathcal{W}^{(n)}$,
  \item a stochastic decoder $G^{(n)}$ which maps $w \in \mathcal{W}^{(n)}$ and $y^n \in \mathcal{Y}^{n}$ to $\hat{x}^n \in \mathcal{\hat{X}}^{n}$ according to $p(\hat{x}^n|w,y^n)$,
\end{itemize}
where  $\mathcal{W}^{(n)}$ is a finite set.
\end{definition}
Let $d: \mathcal{F} \times \hat{\mathcal{X}} \rightarrow [0,\infty)$
  be the single-letter distortion measure\footnote{Note that here $\hat{\mathcal{X}}$ does not denote an alphabet of the reconstruction of $X$, but of the outcome of the function $F(X,Y)$.}. The distortion
  between the value of the function of source sequence and side information and its estimate at the
  decoder is defined as
\begin{align*}
     &d^{(n)}(F^{(n)}(X^n,Y^n),\hat{X}^{n}) \triangleq \frac{1}{n}\sum_{i=1}^{n}d(F(X_i,Y_i),\hat{X}_{i}),
\end{align*}
where $d^{(n)}(\cdot)$ is the distortion function.

\begin{definition}  A rate-distortion-equivocation tuple $(R,D,\triangle) \in \mathbb{R}^{3}_{+}$ is said to be \emph{achievable} if for any $\delta>0$ and all sufficiently large $n$ there exists a $(|\mathcal{W}^{(n)}|,n)$ code such that
\[ \frac{1}{n}\log\big|\mathcal{W}^{(n)}\big|\leq R+\delta, \]
\[ E[d^{(n)}(F^{(n)}(X^n,Y^n), \hat{X}^{n})] \leq D+\delta, \]
\[ \text{and} \quad  \frac{1}{n}H(\hat{X}^{n}|W,Z^n) \geq \triangle-\delta.\]

The \emph{rate-distortion-equivocation region} $\mathcal{R}_{\text{eve}}$ is the set of all achievable tuples.
\end{definition}

\subsection{Result}
\begin{definition}
Let $\mathcal{R}_{\text{in}}^{\text{(eve)}}$ be the set of all tuples $(R,D,\triangle)\in \mathbb{R}^{3}_{+}$ such that
\begin{align}
R &\geq I(X;U|Y) \\
D &\geq E[d(F(X,Y),\hat{X})] \\
\triangle 
& \leq H(\hat{X}|U,Y) + I(\hat{X};Y|T)-I(\hat{X};Z|T) -I(U;Z|T,Y,\hat{X}),
\end{align}
for some joint distributions of form $P_{X,Y,Z}(x,y,z)P_{U|X}(u|x)P_{T|U}(t|u)P_{\hat{X}|U,Y}(\hat{x}|u,y)$ with  $|\mathcal{T}| \leq |\mathcal{X}| +5, |\mathcal{U}| \leq (|\mathcal{X}|+5)(|\mathcal{X}|+4)$.

In addition, let $\mathcal{R}_{\text{out}}^{\text{(eve)}}$ be the same set as $\mathcal{R}_{\text{in}}^{\text{(eve)}}$ except that the equivocation bound is replaced by
\begin{equation}
  \triangle \leq H(\hat{X}|U,Y) + I(V,\hat{X};Y|T) -I(V,\hat{X};Z|T),
\end{equation}
for some joint distributions $P_{X,Y,Z}(x,y,z)P_{U|X}(u|x)P_{T|U}(t|u)P_{V,\hat{X}|U,Y}(v,\hat{x}|u,y)$ where $H(T|V)=H(T|U)=0$.
\end{definition}

\begin{proposition}[Inner and outer bounds]\label{proposition:Chap_ENDUSER_1}
The rate-distortion-equivocation region $\mathcal{R}_{\text{eve}}$ for the problem in Fig. \ref{fig:Chap_ENDUSER_eve} satisfies $\mathcal{R}_{\text{in}}^{\text{(eve)}} \subseteq \mathcal{R}_{\text{eve}} \subseteq \mathcal{R}_{\text{out}}^{\text{(eve)}}$.
\end{proposition}
\begin{proof}
 The proof is given in Appendix \ref{appendix:Chap_ENDUSER_proposition1}. The achievable scheme is based on layered coding and Wyner-Ziv binning in which the former aims to provide some degree of freedom to adapt amount of information accessible to the eavesdropper by utilizing two layers of codewords $T^n$ and $U^n$, and the latter is used to reduce the rate needed for transmission. In addition, we allow for a stochastic decoder where the final reconstruction sequence is generated randomly based on the selected codeword $U^n$ and the side information $Y^n$.
\end{proof}

In the equivocation bound of $\mathcal{R}_{\text{in}}^{\text{(eve)}}$, the first term corresponds to uncertainty of $\hat{X}^n$ due to the use of a stochastic decoder. The difference $I(\hat{X};Y|T)-I(\hat{X};Z|T)$ can be considered as an additional uncertainty due to the fact that the eavesdropper observes $Z^n$, but not $Y^n$ which is used for generating $\hat{X}^n$.
The last mutual information term is related to the leakage of the second layer codeword $U^n$. However, the fact that it is not clear to interpret might be an indication that the bound is not optimal.
From the proof of the outer bound $\mathcal{R}_{\text{out}}^{\text{(eve)}}$, random variable $V$ is related to certain reconstruction symbols and it appears since the reconstruction symbol
depends on the source description and the \emph{whole} side information $Y^n$ (see, e.g., \eqref{eq:random_variable_V} where we cannot simplify further the terms with $\hat{X}^n$ in the conditioning.).

\begin{remark}
We can relate our result to those of other settings where the function $F^{(n)}(X^n,Y^n)=X^n$. For example, the inner bound  $\mathcal{R}_{\text{in}}^{\text{(eve)}}$ can resemble the optimal result of the secure lossless source coding problem considered in \cite{VillardPiantanida}. To obtain the rate-equivocation region, we set $\hat{X}=U=X$ in $\mathcal{R}_{\text{in}}^{\text{(eve)}}$.
\end{remark}

\subsection{Causal Side Information}
Next, we consider the variant of the problem in Fig. \ref{fig:Chap_ENDUSER_eve} where the side information $Y^n$ is available only causally at the decoder. This could be relevant in delay-constrained applications as mentioned in \cite{WeissmanElGamal} and references therein. We consider the following types of reconstructions.
\begin{itemize}
\item \emph{Causal reconstruction}: $\hat{X}_i \sim p(\hat{x}_i|w,y^i,\hat{x}^{i-1})$ for $i=1,\ldots,n$.
\item \emph{Memoryless reconstruction}: $\hat{X}_i \sim p(\hat{x}_i|w,y_i)$ for $i=1,\ldots,n$.
\end{itemize}

\begin{definition}
Let $\mathcal{R}_{\text{in}}^{(\text{eve,causal})}$ be the set of all tuples $(R,D,\triangle)\in \mathbb{R}^{3}_{+}$ such that
\begin{align}
R &\geq I(X;U) \\
D &\geq E[d(F(X,Y),\hat{X})] \\
\triangle &\leq H(\hat{X}|U,Z),
\end{align}
for some joint distributions of the form $P_{X,Y,Z}(x,y,z)P_{U|X}(u|x)P_{\hat{X}|U,Y}(\hat{x}|u,y)$ with $|\mathcal{U}| \leq |\mathcal{X}| +3$.

In addition, let $\mathcal{R}_{\text{out}}^{(\text{eve,causal})}$ be the same set as $\mathcal{R}_{\text{in}}^{(\text{eve,causal})}$ except that the equivocation bound is replaced by
\begin{equation}
  \triangle \leq H(\hat{X}|T,Z),
\end{equation}
for some joint distributions $P_{X,Y,Z}(x,y,z)P_{U|X}(u|x)P_{T|U}(t|u)P_{\hat{X}|U,Y}(\hat{x}|u,y)$ where $H(T|U)=0$.
\end{definition}

\subsubsection{Causal Reconstruction}
\begin{proposition}[Inner and outer bounds]\label{proposition:Chap_ENDUSER_2}
The rate-distortion-equivocation region $\mathcal{R}_{\text{eve}}$ for the problem in Fig. \ref{fig:Chap_ENDUSER_eve} with \emph{causal reconstruction} satisfies the relation $\mathcal{R}_{\text{in}}^{(\text{eve,causal})} \subseteq \mathcal{R}_{\text{eve}} \subseteq \mathcal{R}_{\text{out}}^{(\text{eve,causal})}$.
\end{proposition}
\begin{proof}
Since the side information is only available causally at the decoder, it cannot be used for binning to reduce the rate. The achievable scheme follows that of source coding with causal side information \cite{WeissmanElGamal} with the additional use of a stochastic decoder. The proof is given in Appendix \ref{appendix:Chap_ENDUSER_proposition2}.
\end{proof}

The entropy term in the equivocation bound of $\mathcal{R}_{\text{in}}^{(\text{eve,causal})}$ corresponds to uncertainty of the reconstruction sequence given that the eavesdropper can decode the codeword $U^n$ and has access to the side information $Z^n$.

\subsubsection{Memoryless Reconstruction}
\begin{proposition}[Rate-distortion-equivocation region]\label{proposition:Chap_ENDUSER_3}
The rate-distortion-equivocation region $\mathcal{R}_{\text{eve}}$ for the problem in Fig. \ref{fig:Chap_ENDUSER_eve} with \emph{memoryless reconstruction} is given by $\mathcal{R}_{\text{in}}^{(\text{eve,causal})}$, i.e., $\mathcal{R}_{\text{eve}}= \mathcal{R}_{\text{in}}^{(\text{eve,causal})}$.
\end{proposition}
\begin{proof}
The achievability proof follows the same as in the case of causal reconstruction. As for the converse proof, let $U_i \triangleq W$ which satisfies $U_i-X_i-(Y_i,Z_i)$ and   $\hat{X}_i-(U_i,Y_i)-(X_i,Z_i)$ for all $i=1,\ldots,n$. It then follows that
\begin{align*}
n(R + \delta_n) &\geq H(W) \geq I(X^n;W)\\
&= \sum_{i=1}^n H(X_i) -H(X_i|W,X^{i-1})\\
&\geq \sum_{i=1}^n I(X_i;U_i),
\end{align*}
\begin{align*}
D + \delta_n &\geq E[d^{(n)}(F^{(n)}(X^n,Y^n),\hat{X}^n)] \\
&= \frac{1}{n}\sum_{i=1}^n E[d(F(X_i,Y_i),\hat{X}_i)],
\end{align*}
and
\begin{align*}
n(\triangle - \delta_n) &\leq H(\hat{X}^n|W,Z^n) \\
&\leq \sum_{i=1}^n H(\hat{X}_i|U_i,Z_i).
\end{align*}
The proof ends using the standard time-sharing argument. The cardinality bounds on the sets $\mathcal{U}$ in $R_{\text{in}}^{(\text{eve,causal})}$ can be proved using the support lemma \cite{CsiszarBook} that $\mathcal{U}$ should have $|\mathcal{X}|-1$ elements to preserve $P_{X}$, plus four more for $H(X|U)$, $H(\hat{X}|U,Z)$, $E[d(F(X,Y),\hat{X}]$, and the Markov relation $\hat{X}-(U,Y)-(X,Z)$.
\end{proof}

\begin{remark}\label{remark:Chap_ENDUSER_random_dec_gain}
For the special case where $Y=\varnothing$, the rate-distortion-equivocation region is given by $\mathcal{R}_{\text{in}}^{(\text{eve,causal})}$ with the corresponding set of distributions such that $Y=\varnothing$. We can see that if the decoder is a deterministic mapping, the achievable equivocation rate is zero since the eavesdropper observes everything the decoder does. However, for some positive $D$, by using the stochastic decoder, we can achieve the equivocation rate of $H(\hat{X}|U,Z)$ which can be strictly positive. This shows that there exist cases where stochastic decoder strictly enlarges the rate-distortion-equivocation region.
\end{remark}
\begin{remark}
Proposition \ref{proposition:Chap_ENDUSER_3} resembles the result of the special case in \cite[Corollary 5]{Schieler} where there is no shared secret key.
\end{remark}
\subsection{Special Case: End-user privacy at the encoder}
Fig. \ref{fig:Chap_ENDUSER_eve} includes the setting of end-user privacy at the encoder in Fig. \ref{fig:Chap_ENDUSER_enc} as a special case by setting $Z^n=X^n$ since the source description is a deterministic function of $X^n$.
The above results can readily reduce to the corresponding results for the problems in Fig. \ref{fig:Chap_ENDUSER_enc} as follows.
\begin{itemize}
\item Inner bound: The inner bound for the setting in Fig. \ref{fig:Chap_ENDUSER_enc} is obtained readily from $\mathcal{R}_{\text{in}}^{(\text{eve})}$ by setting $Z=X$ and $T=U$.
\item Inner and outer bounds for \emph{causal reconstruction} are obtained from $\mathcal{R}_{\text{in}}^{(\text{eve,causal})}$ and $\mathcal{R}_{\text{out}}^{(\text{eve,causal})}$ by setting $Z=X$.
\item The rate-distortion-equivocation region for \emph{memoryless reconstruction} is obtained from $\mathcal{R}_{\text{in}}^{(\text{eve,causal})}$ by setting $Z=X$.
\end{itemize}

\section{End-user Privacy at Helper} \label{sec:Chap_ENDUSER_help}

In this section, we consider the setting in Fig. \ref{fig:Chap_ENDUSER_help} where the end-user privacy constraint is imposed at the helper who provides side information $Y^n$ to the decoder. We are interested in how the decoder should utilize the correlated side information in the reconstruction while keeping the reconstruction sequence secret/private from the helper.

\subsection{Problem Formulation}
The problem formulation and definition of the code are similar as before, except that the end-user privacy constraint is now at the helper.

\begin{definition}  A rate-distortion-equivocation tuple $(R,D,\triangle)\in \mathbb{R}^{3}_{+}$ is said to be \emph{achievable} if for any $\delta>0$ and all sufficiently large $n$ there exists a $(|\mathcal{W}^{(n)}|,n)$ code such that
\[ \frac{1}{n}\log\big|\mathcal{W}^{(n)}\big|\leq R+\delta, \]
\[ E[d^{(n)}(F^{(n)}(X^n,Y^n), \hat{X}^n)] \leq D+\delta, \]
\[ \text{and} \quad  \frac{1}{n}H(\hat{X}^n|Y^n) \geq \triangle-\delta.\]

The \emph{rate-distortion-equivocation region} $\mathcal{R}_{\text{helper}}$ is the set of all achievable tuples.
\end{definition}

\subsection{Result}
\begin{definition}
Let $\mathcal{R}_{\text{in}}^{\text{(help)}}$ be the set of all tuples $(R,D,\triangle)\in \mathbb{R}^{3}_{+}$ such that
\begin{align}
R &\geq I(X;U|Y) \\
D &\geq E[d(F(X,Y),\hat{X})] \\
\triangle &\leq H(\hat{X}|U,Y)+ I(X;\hat{X}|Y),
\end{align}
for some joint distributions of the form $P_{X,Y}(x,y)P_{U|X}(u|x)P_{\hat{X}|U,Y}(\hat{x}|u,y)$ with $|\mathcal{U}| \leq |\mathcal{X}| +3$.

In addition, let $\mathcal{R}_{\text{out}}^{\text{(help)}}$ be the same set as $\mathcal{R}_{\text{in}}^{\text{(help)}}$ except that the equivocation bound is replaced by
\begin{equation}
  \triangle \leq H(\hat{X}|U,Y)+ I(X;V,\hat{X}|Y),
\end{equation}
and the joint distributions factorized as $P_{X,Y}(x,y)P_{U|X}(u|x)P_{V,\hat{X}|U,Y}(v,\hat{x}|u,y)$.
\end{definition}

\begin{proposition}[Inner and outer bounds]\label{proposition:Chap_ENDUSER_4}
The rate-distortion-equivocation region $\mathcal{R}_{\text{help}}$ for the problem in Fig. \ref{fig:Chap_ENDUSER_help} satisfies $\mathcal{R}_{\text{in}}^{\text{(help)}} \subseteq \mathcal{R}_{\text{help}} \subseteq \mathcal{R}_{\text{out}}^{\text{(help)}}$.
\end{proposition}

\begin{proof}
The proof is given in Appendix \ref{appendix:Chap_ENDUSER_proposition4} in which the achievable scheme implements Wyner-Ziv type coding with the additional use of a stochastic decoder. We
note that since there is no eavesdropper in this setting, no layering is used in the achievable scheme. As for the outer bound, the presence of random variable $V$ in
$\mathcal{R}_{\text{out}}^{\text{(help)}}$ can be argued similarly as in the proof of Proposition \ref{proposition:Chap_ENDUSER_1}.
\end{proof}

\begin{remark}
One example showing that stochastic decoder can enlarge the rate-distortion-equivocation region is when $Y=X$ in Fig. \ref{fig:Chap_ENDUSER_help}. Since the source is available completely at the decoder, we do not need to send any description over the rate-limited link and the zero rate is achievable. In this case, we have that $\mathcal{R}_{\text{help}}$ is given by the inner bound $\mathcal{R}_{\text{in}}^{\text{(help)}}$ where $X=Y$ and $U=\varnothing$. For any positive $D$, the stochastic decoder could randomly put out a reconstruction sequence that still satisfies the distortion level $D$, and achieve a positive equivocation rate as opposed to the zero equivocation in the case of a deterministic decoder.
\end{remark}
\section{Binary Example}\label{sec:Chap_ENDUSER_example}
In this section, we consider an example illustrating the potential gain from allowing the use of a stochastic decoder. Specifically, we consider the setting in Fig. \ref{fig:Chap_ENDUSER_eve} under memoryless reconstruction and assumptions that  $Z=\emptyset$ and $F(X,Y)=X$. Then we evaluate the corresponding result in Proposition \ref{proposition:Chap_ENDUSER_3}.

Let $\mathcal{X} = \mathcal{\hat{X}}= \{0,1\}$ be binary source and reconstruction alphabets. We assume that the source symbol $X$ is distributed according to Bernoulli(1/2), and side information $Y \in \{0,1,e\}$ is an erased version of the source with an erasure probability $p_e$. The Hamming distortion measure is assumed, i.e., $d(x,\hat{x})=1$ if $x\neq\hat{x}$, and zero otherwise. Inspired by the optimal choice of $U$ in the Wyner-Ziv result \cite{WynerZiv}, we let $U$ be the output of a BSC($p_u$), $p_u \in [0,1/2]$ with input $X$. The reconstruction symbol generated from a stochastic decoder is chosen s.t. $\hat{X} = Y$ if $Y \neq e$, otherwise $\hat{X} \sim P_{\hat{X}|U}$, where $P_{\hat{X}|U}$ is modelled as a BSC($p_2$), $p_2 \in [0,1/2]$. With these assumptions at hand, the inner bound to the rate-distortion-equivocation region in Proposition \ref{proposition:Chap_ENDUSER_3} can be expressed as
\begin{align*}
\mathcal{R}_{\text{in,random}} = \{(R,D,\triangle)| R &\geq 1-h(p_u) \\
D &\geq p_e(p_u\star p_2) \\
\triangle &\leq h(p_u(1-p_e) + p_2p_e)\\
\text{for some} \ p_u,p_2 \in [0,1/2] \},
\end{align*}
where $h(\cdot)$ is a binary entropy function and $a \star b \triangleq a(1-b) + (1-a)b$.

\begin{figure}[t]
    \centering
    \psfrag{del}[][]{\small{$\triangle_{\text{sat}}$}}
    \includegraphics[width=0.65\textwidth]{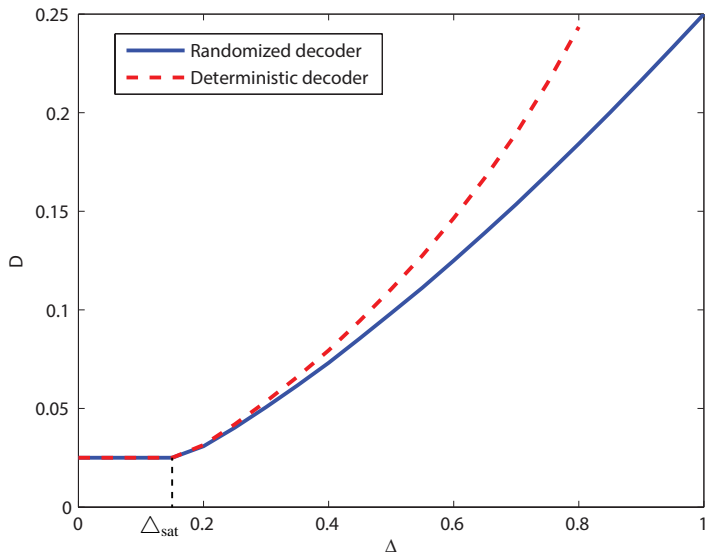}
    \caption{Achievable minimum distortion w.r.t. equivocation for a fixed rate $R=0.7136$, and $p_e=0.5$.}\label{fig:Chap_ENDUSER_example}
\end{figure}

For comparison, we also evaluate the inner bound
for the case of the Wyner-Ziv optimal deterministic decoder by setting $p_2=0$.
We plot the achievable minimum distortion as a function of equivocation rate for a fixed $R=0.7136$, where $p_e=0.5$. Fig. \ref{fig:Chap_ENDUSER_example} shows the tradeoff between achievable minimum distortion and equivocation rate for a fixed rate $R$.  We can see that in general the minimum distortion is sacrificed for a higher equivocation. For the same particular structure of $P_{U|X}$ and the given deterministic decoder in this setting, it shows that, for a given rate $R$ and distortion $D$, a higher equivocation rate $\triangle$ can be achieved by using a stochastic decoder.\footnote{Here we only evaluate and compare inner bounds on the rate-distortion-equivocation regions to illustrate a potential gain of allowing the use of a stochastic decoder.}
As for the low equivocation region, we observe a saturation of distortion because the minimum distortion is limited by the rate. The value $\triangle_{\text{sat}}$ at which the minimum distortion cannot be lowered by decreasing $\triangle$ can be specified as $\triangle_{\text{sat}}= h((1-p_e)h^{-1}(1-R))$, and the corresponding $D_{\text{min}}(R,\triangle_{\text{sat}}) = p_e h^{-1}(1-R)$ is the minimum distortion according to the Wyner-Ziv rate-distortion function. It could also be interesting to see how good the inner bounds are by evaluating the outer bound result. However, it involves the optimization over an auxiliary random variable which is not straightforward and will be left for future work.

\emph{Special case:} In the special case where $Y=\emptyset$, the gain can be shown as follows (cf. Remark \ref{remark:Chap_ENDUSER_random_dec_gain}). If the decoder is a deterministic mapping, the achievable equivocation rate is always zero since the eavesdropper is as strong as the decoder. The corresponding distortion-rate function for this example is given by $D \geq h^{-1}(1-R)$ \cite[Ch.3]{ElGamalKim}. However, by using a stochastic decoder as above, we can achieve $D \geq h^{-1}(1-R)\star h^{-1}(\triangle)$ (by letting $p_e=1$ in $\mathcal{R}_{\text{in,random}}$). For $D=h^{-1}(1-R) \star c$, where $c \in (0,1/2]$, we can achieve strictly positive equivocation rate  $h(c)$.

\section{Conclusion}
In this paper, we introduced a new privacy metric (end-user privacy constraint) in the problems of lossy source coding with side information. We considered several problems where the end-user privacy constraint is imposed at different nodes, namely the eavesdropper, the encoder, and the helper. Since the goal of end-user privacy is to protect the reconstruction sequence generated at the decoder against any unwanted inferences, we allow the decoder mapping to be a random mapping, and it was shown by example that there exist cases where a stochastic decoder strictly enlarges the rate-distortion-equivocation region as compared to the one derived for deterministic decoders. In general, characterizing the complete rate-distortion-equivocation region for the setting with end-user privacy is difficult since conditioned on the source description, the reconstruction process is not necessarily memoryless. As seen in a special case of end-user privacy at the eavesdropper, when we restrict the reconstruction symbol to depend only on the source description and the current side information symbol, the complete rate-distortion-equivocation region can be given.
\appendices

\section{Proof of Proposition \ref{proposition:Chap_ENDUSER_1}} \label{appendix:Chap_ENDUSER_proposition1}
The inner bound proofs for the rate and distortion constraints follow from the coding scheme which utilizes layered coding and Wyner-Ziv binning. That is, we have two layers of codewords $T^n$ and $U^n$ forming the codebook, and after encoding, only the bin indices of the chosen codewords are transmitted to the decoder. Also, instead of using the deterministic function at the decoder, we allow stochastic decoder to generate the reconstruction sequence, i.e., the decoder puts out $\hat{X}^n$, where $\hat{X}_i \sim P_{\hat{X}|U,Y}$ for each $i=1,\ldots,n$. The outline of the proof is given below.

Fix $P_{U|X}, P_{T|U}$, and $P_{\hat{X}|U,Y}$. Randomly and independently generate $2^{n(I(X;T) + \delta_{\epsilon})}$ $t^n(j)$ sequences, each i.i.d. according to $\prod_{i=1}^nP_T(t_i)$, $j \in [1:2^{n(I(X;T) + \delta_{\epsilon})}]$. Then distribute them uniformly at random into $2^{n(I(X;T|Y) + 2\delta_{\epsilon})}$ equal-sized bins $b_T(w_1)$, $w_1 \in [1:2^{nI(X;T|Y) + 2\delta_{\epsilon}}]$. For each $j$, randomly and conditionally independently generate $2^{n(I(X;U|T) + \delta_{\epsilon})}$ $u^n(j,k)$ sequences, each i.i.d. according to $\prod_{i=1}^nP_{U|T}(u_i|t_i)$, $k \in [1:2^{n(I(X;U|T) + \delta_{\epsilon})}]$, and distribute these sequences uniformly at random into $2^{n(I(X;U|T,Y) + 2\delta_{\epsilon})}$ equal-sized bins $b_U(j,w_2)$, $w_2 \in [1:2^{nI(X;U|T,Y) + 2\delta_{\epsilon}}]$. For encoding, the encoder looks for $t^n(j)$ and $u^n(j,k)$ jointly typical with $x^n$. With high probability, it will find such codewords and then send the corresponding bin indices $w_1$ and $w_2$ to the decoder. The total rate is thus equal to $I(X;T|Y)+I(X;U|T,Y) + 4\delta_{\epsilon} = I(X;U|Y) + 4\delta_{\epsilon}$. Based on the received bin indices, the decoder, with high probability, will find the unique sequences $t^n(j) \in b_T(w_1)$ and $u^n(j,k) \in b_U(j,w_2)$ such that they are jointly typical with $y^n$. Then it puts out $\hat{x}^n$ where $\hat{x}_i$ is randomly generated according to $P_{\hat{X}|U,Y}(\hat{x}_i|u_i,y_i)$, $i=1,\ldots,n$.

Let $T^n(J)$ and $U^n(J,K)$ be the codewords chosen at the encoder, and $W_1$ and $W_2$ be the corresponding bin indices of the bins which $T^n(J)$ and $U^n(J,K)$ belong to. Then $W_1$ and $W_2$ are functions of $J$ and $K$. Since the tuple $(X^n,T^n(J),U^n(J,K),Y^n,\hat{X}^n) \in \mathcal{T}_{\epsilon}^{(n)}(X,T,U,Y,\hat{X})$ with high probability, it can be shown that the distortion constraint is satisfied if $E[d(F(X,Y),\hat{X})] \leq D$.

Next, we give a sketch of the proof for the equivocation constraint. Let $\mathcal{C}_n$ be a random variable representing the randomly chosen codebook.
By Fano's inequality, we have that for any $\mathcal{C}_n = \mathfrak{C}_n$, $H(J,K|W_1,W_2,Y^n,\mathcal{C}_n = \mathfrak{C}_n) \leq 1 + \mathrm{Pr}(\mathcal{E}) \log(|\mathcal{J}||\mathcal{K}|)$, where $\mathrm{Pr}(\mathcal{E})$ is the probability that $(J,K)$ cannot be identified from $(W_1,W_2,Y^n)$.  From the decoding process, we have that $\mathrm{Pr}(\mathcal{E}) \rightarrow0$ as $n\rightarrow \infty$. Then, it follows that
\begin{align}
  \frac{1}{n} H(J,K|W_1,W_2,Y^n,\mathcal{C}_n) &= \sum_{\mathfrak{C}_n}p(\mathfrak{C}_n)\frac{1}{n} H(J,K|W_1,W_2,Y^n,\mathcal{C}_n = \mathfrak{C}_n)\nonumber \\
  &\leq \sum_{\mathfrak{C}_n}p(\mathfrak{C}_n) (\frac{1}{n} + \mathrm{Pr}(\mathcal{E})\frac{1}{n} \log(|\mathcal{J}||\mathcal{K}|))\nonumber \\
   &\leq \epsilon_n, \label{eq:Chap_ENDUSER_Fano1}
\end{align}
where $\epsilon_n \rightarrow0$ as $n\rightarrow \infty$.

Since $(X^n,T^n(J),U^n(J,K),Y^n,\hat{X}^n) \in \mathcal{T}_{\epsilon}^{(n)}$ with high probability, we also have the following lemmas.
\begin{lemma} \label{lemma:Chap_ENDUSER_1}
The following bound holds, $H(X^n|J,K,Y^n,Z^n,\mathcal{C}_n) \leq n[H(X|T,U,Y,Z)+\delta_{\epsilon}]$.
\end{lemma}
\begin{proof}
The proof is given in Appendix \ref{appendix:Chap_ENDUSER_Lemma1}.
\end{proof}
\begin{lemma} \label{lemma:Chap_ENDUSER_2}
If $\hat{X}^n \sim \prod_{i=1}^n P_{\hat{X}|U,Y}(\hat{x}_i|u_i,y_i)$, we have that $H(\hat{X}^n|U^n,Y^n,\mathcal{C}_n) \geq n[H(\hat{X}|U,Y)-\delta_{\epsilon}]$.
\end{lemma}
\begin{proof}
The proof is given in Appendix \ref{appendix:Chap_ENDUSER_Lemma2}.
\end{proof}

The equivocation averaged over all codebooks can be bounded as follows.
\begin{align*}
&H(\hat{X}^n|W_1,W_2,Z^n,\mathcal{C}_n) \\
&= H(\hat{X}^n|J,K,Y^n,Z^n,\mathcal{C}_n)+ I(\hat{X}^n;J,K,Y^n|W_1,W_2,Z^n,\mathcal{C}_n)\\
&\overset{(a)}{\geq} H(\hat{X}^n|U^n,Y^n,\mathcal{C}_n) + I(\hat{X}^n;Y^n|W_1,W_2,Z^n,\mathcal{C}_n)  + H(J,K|W_1,W_2,Y^n,Z^n,\mathcal{C}_n)-n\epsilon_n\\
&= H(\hat{X}^n|U^n,Y^n,\mathcal{C}_n) + H(Y^n,Z^n) + H(W_1,W_2|Y^n,Z^n,\mathcal{C}_n) -H(W_1,W_2,Z^n|\mathcal{C}_n) \\ &\qquad -H(Y^n|W_1,W_2,Z^n,\hat{X}^n,\mathcal{C}_n)+H(J,K|W_1,W_1,Y^n,Z^n,\mathcal{C}_n)-n\epsilon_n\\
&= H(\hat{X}^n|U^n,Y^n,\mathcal{C}_n) + H(Y^n,Z^n) + H(J,K|Y^n,Z^n,\mathcal{C}_n) -H(W_1,W_2,Z^n|\mathcal{C}_n)\\ &\qquad -H(Y^n|W_1,W_2,Z^n,\hat{X}^n,\mathcal{C}_n)-n\epsilon_n\\
&\overset{(b)}{\geq} H(\hat{X}^n|U^n,Y^n) + H(Y^n,Z^n) + I(J,K;X^n|Y^n,Z^n,\mathcal{C}_n)  -H(W_1|\mathcal{C}_n)\\ &\qquad -H(W_2|\mathcal{C}_n)-H(Z^n|W_1,\mathcal{C}_n)-H(Y^n|W_1,Z^n,\hat{X}^n,\mathcal{C}_n)-n\epsilon_n\\
&= H(\hat{X}^n|U^n,Y^n) + H(Y^n,Z^n) + I(J,K;X^n|Y^n,Z^n,\mathcal{C}_n) -H(W_1|\mathcal{C}_n)-H(W_2|\mathcal{C}_n)\\ &\qquad-H(Z^n|J,\mathcal{C}_n)-I(Z^n;J|W_1,\mathcal{C}_n)- H(Y^n|J,Z^n,\hat{X}^n,\mathcal{C}_n) -I(Y^n;J|W_1,Z^n,\hat{X}^n,\mathcal{C}_n)-n\epsilon_n\\
&\geq H(\hat{X}^n|U^n,Y^n) + H(Y^n,Z^n) + I(J,K;X^n|Y^n,Z^n,\mathcal{C}_n) -H(J|\mathcal{C}_n)\\ &\qquad -H(W_2|\mathcal{C}_n)-H(Z^n|J,\mathcal{C}_n)- H(Y^n|J,Z^n,\hat{X}^n,\mathcal{C}_n)-n\epsilon_n\\
&\overset{(c)}{\geq} n[H(\hat{X}|U,Y)+ H(Y,Z) + I(X;T,U|Y,Z)-I(X;T)-I(X;U|T,Y)\\ &\qquad -H(Z|T)-H(Y|T,Z,\hat{X}) -\delta_{\epsilon}'- \epsilon_n] \\
&\overset{(d)}{=} n[H(\hat{X}|U,Y) + I(\hat{X};Y|T)-I(\hat{X};Z|T) -I(U;Z|T,Y,\hat{X})-\delta_{\epsilon}'']\\
& \geq n[\triangle -\delta_{\epsilon}''],
\end{align*}
if $\triangle \leq H(\hat{X}|U,Y) + I(\hat{X};Y|T)-I(\hat{X};Z|T) -I(U;Z|T,Y,\hat{X})$,
 where $(a)$ follows from, conditioned on the codebook, we have the Markov chain $\hat{X}^n-(U^n(J,K),Y^n)-(J,K,Z^n)$, and from Fano's inequality in \eqref{eq:Chap_ENDUSER_Fano1}, $(b)$ follows the from the Markov chain $\hat{X}^n-(U^n(J,K),Y^n)-\mathcal{C}_n$, from $(J,K)$ is a function of $X^n$, and that conditioning reduces entropy, $(c)$ follows from the codebook generation and from bounding the term $H(\hat{X}^n|U^n,Y^n)$ where $\hat{X}^n \sim \prod_{i=1}^n P_{\hat{X}|U,Y}$ as in Lemma \ref{lemma:Chap_ENDUSER_2}, and the terms $H(X^n|J,K,Y^n,Z^n,\mathcal{C}_n)$, $H(Z^n|J,\mathcal{C}_n)$, and $H(Y^n|J,Z^n,\hat{X}^n,\mathcal{C}_n)$ for which the proofs follow similarly as that of Lemma \ref{lemma:Chap_ENDUSER_1}, and $(d)$ from the Markov chains $T-U-X-(Y,Z)$ and $\hat{X}-(U,Y)-(X,Z,T)$.

The cardinality bounds on the sets $\mathcal{T}$ and $\mathcal{U}$ in $\mathcal{R}_{\text{in}}^{(\text{eve})}$ can be proved using the support lemma \cite{CsiszarBook} and is given in Appendix \ref{appendix:Chap_ENDUSER_cardinality}.

As for the outer bound, let $T_i \triangleq (W,Z^{i-1},Y_{i+1}^n)$, $U_i \triangleq (W,Z^{i-1},Y^{n \setminus i})$ and $V_i \triangleq (W,Z^{i-1},Y_{i+1}^n,\hat{X}^{n \setminus i})$ which satisfy $T_i-U_i-X_i-(Y_i,Z_i)$ and $(V_i,\hat{X}_i)-(U_i,Y_i)-(X_i,Z_i,T_i)$ for all $i=1,\ldots,n$. The outer bound proof for the rate and distortion constraints follows similarly as that of the Wyner-Ziv problem with the exception of the part related to stochastic decoder. That is, we have
\begin{align*}
n(R + \delta_n) &\geq H(W) \geq I(X^n,Z^n;W|Y^n)\\
&= \sum_{i=1}^n H(X_i,Z_i|Y_i) -H(X_i,Z_i|W,X^{i-1},Z^{i-1},Y^n)\\
&\overset{(a)}{\geq} \sum_{i=1}^n H(X_i,Z_i|Y_i) -H(X_i,Z_i|U_i,Y_i)\\
&\geq \sum_{i=1}^n I(X_i;U_i|Y_i),
\end{align*}
where $(a)$ follows from the definition of $U_i$ and that conditioning reduces entropy, and
\begin{align*}
D + \delta_n &\geq E[d^{(n)}(F^{(n)}(X^n,Y^n),\hat{X}^n)] \\
&= \frac{1}{n}\sum_{i=1}^n E[d(F(X_i,Y_i),\hat{X}_i)].
\end{align*}

The equivocation bound follows below.
\begin{align}
n(\triangle - \delta_n) &\leq H(\hat{X}^n|W,Z^n) = H(\hat{X}^n|W) -I(\hat{X}^n;Z^n|W) \nonumber \\
&= H(\hat{X}^n|W,Y^n) + I(\hat{X}^n;Y^n|W) -I(\hat{X}^n;Z^n|W) \nonumber \\
&\leq \sum_{i=1}^n H(\hat{X}_i|W,Y^n) + H(Y_i|W,Y_{i+1}^n) -H(Y_i|W,Y_{i+1}^n,\hat{X}^n) \nonumber \\ &\qquad -H(Z_i|W,Z^{i-1})+ H(Z_i|W,Z^{i-1},\hat{X}^n) \label{eq:random_variable_V} \\
&\overset{(a)}{\leq} \sum_{i=1}^n H(\hat{X}_i|W,Y^n,Z^{i-1}) - I(Y_i;W,Y_{i+1}^n) +I(Y_i;\hat{X}_i) + I(Y_i; W,Y_{i+1}^n,\hat{X}^{n \setminus i}|\hat{X}_i) \nonumber\\ &\qquad +I(Z_i;W,Z^{i-1})-I(Z_i;\hat{X}_i)  -I(Z_i;W,Z^{i-1},\hat{X}^{n \setminus i}|\hat{X}_i)\nonumber \\
&\overset{(b)}{=} \sum_{i=1}^n H(\hat{X}_i|W,Y^n,Z^{i-1}) - I(Y_i;W,Z^{i-1},Y_{i+1}^n) +I(Y_i;\hat{X}_i) + I(Y_i; W,Z^{i-1},Y_{i+1}^n,\hat{X}^{n \setminus i}|\hat{X}_i) \nonumber\\ &\qquad+I(Z_i;W,Z^{i-1},Y_{i+1}^n) -I(Z_i;\hat{X}_i) -I(Z_i;W,Z^{i-1},Y_{i+1}^n,\hat{X}^{n \setminus i}|\hat{X}_i)\nonumber\\
&\overset{(c)}{=} \sum_{i=1}^n H(\hat{X}_i|U_i,Y_i) - I(Y_i;T_i) +I(Y_i;\hat{X}_i)+ I(Y_i; T_i,V_i|\hat{X}_i) \nonumber\\ &\qquad +I(Z_i;T_i) -I(Z_i;\hat{X}_i)  -I(Z_i;T_i,V_i|\hat{X}_i)\nonumber\\
&= \sum_{i=1}^n H(\hat{X}_i|U_i,Y_i) +I(Y_i;V_i,\hat{X}_i|T_i) -I(Z_i;V_i,\hat{X}_i|T_i)\nonumber,
\end{align}
where $(a)$ follows from the Markov chain $\hat{X}_i-(W,Y^n)-Z^{i-1}$, $(b)$ follows from the Csisz\'{a}r's sum identity, $\sum_{i=1}^n I(Y_i;Z^{i-1}|W,Y_{i+1}^n) - I(Z_i;Y_{i+1}^n|W,Z^{i-1}) =0$ and  $\sum_{i=1}^n I(Y_i;Z^{i-1}|W,\hat{X}^n,Y_{i+1}^n) - I(Z_i;Y_{i+1}^n|W,\hat{X}^n,Z^{i-1}) =0$, and $(c)$ follows from the definitions of $T_i$, $U_i$ and $V_i$.

Note that from the definitions of $T_i$, $U_i$, and $V_i$, we have that $T_i$ is a function of $U_i$ or $V_i$. So we can further restrict the set of joint distributions to satisfy $H(T_i|U_i)=H(T_i|V_i)=0$. The proof ends using the standard time-sharing argument.

\section{Proof of Lemma \ref{lemma:Chap_ENDUSER_1}} \label{appendix:Chap_ENDUSER_Lemma1}
Let $E$ be the binary random variable taking value $0$ if $(X^n,T^n(J),U^n(J,K),Y^n,Z^n) \in \mathcal{T}_{\epsilon}^{(n)}$, and $1$ otherwise. Since $(X^n,T^n(J),U^n(J,K),Y^n,Z^n) \in \mathcal{T}_{\epsilon}^{(n)}$ with high probability, we have $\text{Pr}(E=1) \leq \delta_{\epsilon}$. It follows that
\begin{align*}
 &H(X^n|J,K,Y^n,Z^n,\mathcal{C}_n)\leq H(X^n|T^n(J),U^n(J,K),Y^n,Z^n)\\
 &\leq H(X^n|T^n,U^n,Y^n,Z^n,E) + H(E)\\
 &\leq \text{Pr}(E=0) H(X^n|T^n,U^n,Y^n,Z^n,E=0) + \text{Pr}(E=1) H(X^n|T^n,U^n,Y^n,Z^n,E=1) + h(\delta_{\epsilon})\\
 &\leq H(X^n|T^n,U^n,Y^n,Z^n,E=0) +\delta_{\epsilon} H(X^n)+ h(\delta_{\epsilon}) \\
 &\leq H(X^n|T^n,U^n,Y^n,Z^n,E=0) + n\delta_{\epsilon} \log|\mathcal{X}| + h(\delta_{\epsilon})\\
&= \sum_{(t^n,u^n,y^n,z^n) \in \mathcal{T}_{\epsilon}^{(n)}} p(t^n,u^n,y^n,z^n|E=0) H(X^n|T^n=t^n,U^n=u^n,Y^n=y^n,Z^n=z^n,E=0) \\ &\qquad + n\delta_{\epsilon} \log|\mathcal{X}| + h(\delta_{\epsilon})\\
&\leq \sum_{(t^n,u^n,y^n,z^n) \in \mathcal{T}_{\epsilon}^{(n)}} p(t^n,u^n,y^n,z^n|E=0) \log|\mathcal{T}_{\epsilon}^{(n)}(X|t^n,u^n,y^n,z^n)| + n\delta_{\epsilon} \log|\mathcal{X}| + h(\delta_{\epsilon})\\
 &\leq n[H(X|T,U,Y,Z)+\delta_{\epsilon}'],
\end{align*}
where $h(\cdot)$ is the binary entropy function, and the last inequality follows from the property of joint typical set \cite{ElGamalKim} with $\delta_{\epsilon}, \delta_{\epsilon}' \rightarrow 0$ as $\epsilon \rightarrow 0$, and $\epsilon \rightarrow 0$ as $n \rightarrow \infty$.

\section{Proof of Lemma \ref{lemma:Chap_ENDUSER_2}} \label{appendix:Chap_ENDUSER_Lemma2}
Consider $H(\hat{X}^n|U^n,Y^n,\mathcal{C}_n)$ where $\hat{X}^n$ is distributed i.i.d. $\sim P_{\hat{X}|U,Y}$. Note that $\hat{X}^n-(U^n,Y^n)-\mathcal{C}_n$. It then follows that
\begin{align*}
H(\hat{X}^n|U^n,Y^n) &= \sum_{(u^n,y^n) \in \mathcal{T}_{\epsilon}^{(n)}} p(u^n,y^n)H(\hat{X}^n|U^{n}=u^{n},Y^n=y^n)  \\ &\qquad  + \sum_{(u^n,y^n) \notin \mathcal{T}_{\epsilon}^{(n)}} p(u^n,y^n) H(\hat{X}^n|U^{n}=u^{n},Y^n=y^n) \\
& \overset{(a)}{\geq} \sum_{(u^n,y^n) \in \mathcal{T}_{\epsilon}^{(n)}} p(u^n,y^n) \sum_{i=1}^{n} H(\hat{X}_i|U_{i}=u_{i},Y_i=y_i)\\
& = \sum_{(u^n,y^n) \in \mathcal{T}_{\epsilon}^{(n)}} p(u^n,y^n) \sum_{a \in \mathcal{U},b \in \mathcal{Y}} N(a,b|u^n,y^n)H(\hat{X}|U=a,Y=b)\\
& \overset{(b)}{\geq} \sum_{(u^n,y^n) \in \mathcal{T}_{\epsilon}^{(n)}} p(u^n,y^n) \sum_{a \in \mathcal{U},b \in \mathcal{Y}} n p(a,b)(1-\epsilon)H(\hat{X}|U=a,Y=b)\\
& \geq n (H(\hat{X}|U,Y) -\delta_{\epsilon}),
\end{align*}
where $(a)$ follows from memoryless property of $P_{\hat{X}^n|U^n,Y^n}$, and $(b)$ follows from the definition of joint typical set with $\delta_{\epsilon} \rightarrow 0$ as $\epsilon \rightarrow 0$, and $\epsilon \rightarrow 0$ as $n \rightarrow \infty$.

\section{Proof of Proposition \ref{proposition:Chap_ENDUSER_2}} \label{appendix:Chap_ENDUSER_proposition2}
The inner bound proof for the rate and distortion constraints follows that of source coding with causal side information \cite{WeissmanElGamal} with the additional use of a stochastic decoder. Since the side information is only available causally at the decoder, it cannot be used for binning to reduce the rate. Here, we just use the rate-distortion code with codewords $U^n$. The decoder then generates $\hat{X}^n$, where $\hat{X}_i \sim P_{\hat{X}|U,Y}$ for $i=1,\ldots,n$. The proof of equivocation constraint is given below. It is different from the non-causal case in that the scheme does not utilize binning. Here $W$ denotes the index of codeword $U^n$. The equivocation averaged over all codebooks can be bounded as follows.
\begin{align*}
&H(\hat{X}^n|W,Z^n,\mathcal{C}_n)  = H(\hat{X}^n|W,Z^n,Y^n,\mathcal{C}_n) + I(\hat{X}^n;Y^n|W,Z^n,\mathcal{C}_n)\\
&= H(\hat{X}^n|W,Y^n,Z^n,\mathcal{C}_n) + H(Y^n|W,Z^n,\mathcal{C}_n)- H(Y^n|W,Z^n,\hat{X}^n,\mathcal{C}_n)\\
&\overset{(a)}{=} H(\hat{X}^n|U^n(W),Y^n,\mathcal{C}_n) + H(Y^n,Z^n)+H(W|Y^n,Z^n,\mathcal{C}_n) -H(W|\mathcal{C}_n)\\ &\qquad -H(Z^n|W,\mathcal{C}_n)- H(Y^n|W,Z^n,\hat{X}^n,\mathcal{C}_n)\\
&\overset{(b)}{\geq} n[\underbrace{H(\hat{X}|U,Y) + H(Y,Z)- I(X;U)}_{\triangleq P} - \delta_{\epsilon}'] +H(W|Y^n,Z^n,\mathcal{C}_n)-H(Z^n|W,\mathcal{C}_n)\\ &\qquad - H(Y^n|W,Z^n,\hat{X}^n,\mathcal{C}_n)\\
& \overset{(c)}{=} n[P -\delta_{\epsilon}'] + I(W;X^n|Y^n,Z^n,\mathcal{C}_n)-H(Z^n|W,\mathcal{C}_n) - H(Y^n|W,Z^n,\hat{X}^n,\mathcal{C}_n)\\
& = n[P -\delta_{\epsilon}'] + H(X^n|Y^n,Z^n)-H(X^n|W,Y^n,Z^n,\mathcal{C}_n)-H(Z^n|W,\mathcal{C}_n)\\ &\qquad - H(Y^n|W,Z^n,\hat{X}^n,\mathcal{C}_n)]\\
&\overset{(d)}{\geq} n[H(\hat{X}|U,Y) + H(Y,Z)-I(X;U) + H(X|Y,Z) -H(X|U,Y,Z) \\ &\qquad -H(Z|U)-H(Y|U,Z,\hat{X})-\delta_{\epsilon}'']\\
&\overset{(e)}{=} n[H(\hat{X}|U,Z)-\delta_{\epsilon}'']\\
&\geq n[\triangle -\delta_{\epsilon}''],
\end{align*}
if $\triangle \leq H(\hat{X}|U,Z)$, where $(a)$ follows from, conditioned on the codebook, we have the Markov chain $\hat{X}^n-(U^n,Y^n)-(W,Z^n)$, $(b)$ follows from the Markov chain $\hat{X}^n-(U^n(J,K),Y^n)-\mathcal{C}_n$, from bounding the term $H(\hat{X}^n|U^n,Y^n)$ where $\hat{X}^n \sim \prod_{i=1}^n P_{\hat{X}|U,Y}$ as in Lemma \ref{lemma:Chap_ENDUSER_2} and from the codebook generation, $(c)$ follows since $W$ is a function of $X^n$, $(d)$ follows from bounding the terms $H(X^n|W,Y^n,Z^n,\mathcal{C}_n)$, $H(Z^n|W,\mathcal{C}_n)$, and $H(Y^n|W,Z^n,\hat{X}^n,\mathcal{C}_n)$ similarly as that of Lemma \ref{lemma:Chap_ENDUSER_1}, and $(e)$ follows from the Markov chains $U-X-(Y,Z)$ and $\hat{X}-(U,Y)-(X,Z)$.

The cardinality bound on the set $\mathcal{U}$ in $\mathcal{R}_{\text{in}}^{(\text{eve,causal})}$ can be proved using the support lemma that $\mathcal{U}$ should have $|\mathcal{X}|-1$ elements to preserve $P_{X}$, plus four more for $H(X|U)$, $H(\hat{X}|U,Z)$, $E[d(F(X,Y),\hat{X}]$, and the Markov relation $\hat{X}-(U,Y)-(X,Z)$.

For the outer bound proof, let $U_i \triangleq (W,Y^{i-1},\hat{X}^{i-1})$, $T_i \triangleq (W,\hat{X}^{i-1})$  which satisfy $T_i-U_i-X_i-(Y_i,Z_i)$ and $\hat{X}_i-(U_i,Y_i)-(X_i,Z_i,T_i)$ for all $i=1,\ldots,n$. It then follows that
\begin{align*}
n(R + \delta_n) &\geq H(W) \\
&\geq I(X^n;W)\\
&= \sum_{i=1}^n H(X_i) -H(X_i|W,X^{i-1})\\
&\overset{(a)}{=} \sum_{i=1}^n H(X_i) -H(X_i|W,X^{i-1},Y^{i-1},\hat{X}^{i-1})\\
&\geq \sum_{i=1}^n I(X_i;U_i),
\end{align*}
where $(a)$ follows from the Markov chain $X_i-(W,X^{i-1})-(Y^{i-1},\hat{X}^{i-1})$.
And
\begin{align*}
D + \delta_n &\geq E[d^{(n)}(F^{(n)}(X^n,Y^n),\hat{X}^n)] \\
&= \frac{1}{n}\sum_{i=1}^n E[d(F(X_i,Y_i),\hat{X}_i)].
\end{align*}
And
\begin{align*}
n(\triangle - \delta_n) &\leq H(\hat{X}^n|W,Z^n) \\
&= \sum_{i=1}^n H(\hat{X}_i|T_i,Z_i).
\end{align*}
Note that from the definitions of $T_i$ and $U_i$, we have that $T_i$ is a function of $U_i$. So we can further restrict the set of joint distributions to satisfy $H(T_i|U_i)=0$. The proof ends using the standard time-sharing argument.

\section{Proof of Proposition \ref{proposition:Chap_ENDUSER_4}} \label{appendix:Chap_ENDUSER_proposition4}
The inner bound proof for the rate and distortion constraints follows from the scheme that implements Wyner-Ziv type coding with the additional use of a stochastic decoder. We only give a sketch of the proof of equivocation constraint here. The equivocation averaged over all codebooks can be bounded as follows.
\begin{align*}
H(\hat{X}^n|Y^n,\mathcal{C}_n) &= H(\hat{X}^n|X^n,Y^n,\mathcal{C}_n)+ I(\hat{X}^n;X^n,|Y^n,\mathcal{C}_n)\\
&\overset{(a)}{=} H(\hat{X}^n|X^n,U^n,Y^n,\mathcal{C}_n)+ I(\hat{X}^n;X^n,|Y^n,\mathcal{C}_n)\\
&\overset{(b)}{\geq} n[H(\hat{X}|U,Y)+ H(X|Y)- H(X|Y,\hat{X})-\delta_{\epsilon}'] \\
&= n[H(\hat{X}|U,Y)+ I(X;\hat{X}|Y)-\delta_{\epsilon}']\\
& \geq n[\triangle -\delta_{\epsilon}'],
\end{align*}
if $\triangle \leq H(\hat{X}|U,Y)+ I(X;\hat{X}|Y)$,
where $(a)$ follows since $U^n$ is a function of $X^n$, and $(b)$ follows from the Markov chain $\hat{X}^n- (U^n,Y^n)-(X^n,\mathcal{C}_n)$, and from  bounding the terms $H(X^n|Y^n,\hat{X}^n,\mathcal{C}_n)$ and $H(\hat{X}^n|U^n,Y^n)$ similarly as in Lemmas \ref{lemma:Chap_ENDUSER_1} and \ref{lemma:Chap_ENDUSER_2}, respectively.

The cardinality bound on the set $\mathcal{U}$ in $\mathcal{R}_{\text{in}}^{(\text{help})}$ can be proved using the support lemma that $\mathcal{U}$ should have $|\mathcal{X}|-1$ elements to preserve $P_{X}$, plus four more for $H(X|U,Y)$, $H(\hat{X}|U,Y)$, $E[d(F(X,Y),\hat{X}]$, and the Markov relation $\hat{X}-(U,Y)-X$.

The outer bound proof for equivocation constraint is as follows. Let $U_i \triangleq (W,X^{i-1},Y^{n \setminus i})$ and $V_i \triangleq (X^{i-1},Y^{n \setminus i}, \hat{X}^{n \setminus i})$ which satisfy $U_i-X_i-Y_i$ and $(V_i,\hat{X}_i)-(U_i,Y_i)-X_i$ for all $i=1,\ldots,n$. It follows that
\begin{align*}
n(\triangle - \delta_n) &\leq H(\hat{X}^n|Y^n) \\
&= H(\hat{X}^n|X^n,Y^n)+ I(\hat{X}^n;X^n,|Y^n)\\
&\overset{(a)}{=} H(\hat{X}^n|X^n,W,Y^n)+ I(\hat{X}^n;X^n,|Y^n)\\
&= \sum_{i=1}^n H(\hat{X}_i|W,\hat{X}^{i-1},X^n,Y^n)+ H(X_i|Y_i) - H(X_i|X^{i-1},Y^n,\hat{X}^n)\\
&\overset{(b)}{\leq} \sum_{i=1}^n H(\hat{X}_i|U_i,Y_i)+H(X_i|Y_i) - H(X_i|V_i,\hat{X}_i,Y_i)\\
&= \sum_{i=1}^n H(\hat{X}_i|U_i,Y_i)+I(X_i;V_i,\hat{X}_i|Y_i),
\end{align*}
where $(a)$ follows from the deterministic encoder, $(b)$ follows from the definition of $U_i,V_i$. The proof ends using the standard time-sharing argument.

\section{Cardinality Bounds of The Sets $\mathcal{T}$ and $\mathcal{U}$ in Proposition \ref{proposition:Chap_ENDUSER_1}} \label{appendix:Chap_ENDUSER_cardinality}
Consider the expression of $\mathcal{R}_{\text{in}}^{(\text{eve})}$ in Proposition \ref{proposition:Chap_ENDUSER_1}:
\begin{align*}
R &\geq I(X;U|Y) \\
D &\geq E[d(F(X,Y),\hat{X})] \\
\triangle
& \leq H(\hat{X}|U,Y) + I(\hat{X};Y|T)-I(\hat{X};Z|T) -I(U;Z|T,Y,\hat{X}),
\end{align*}
for some $U \in \mathcal{U}$, $T \in \mathcal{T}$ such that $T-U-X-(Y,Z)$ and $\hat{X}-(U,Y)-(X,Z,T)$ form Markov chains.

We can rewrite some mutual information terms in the rate and equivocation expressions above and get
\begin{align*}
     R &\geq H(X|Y)-H(X,Y|U)+H(Y|U),\\
     \triangle &\leq H(\hat{X},Y|U)-H(Y|U)+ I(Y;\hat{X},Z|T)-H(Z|T)+H(Z,Y|U)-H(Y|U).
\end{align*}
We will prove that the random variables $T$ and $U$ may be replaced by new ones, satisfying $|\mathcal{T}| \leq  |\mathcal{X}|+5$, $|\mathcal{U}| \leq  (|\mathcal{X}|+5)(|\mathcal{X}|+4)$, and preserving the terms $H(X,Y|U)-H(Y|U)$, $H(\hat{X},Y|U)-H(Y|U)$, $H(Z,Y|U)-H(Y|U)$,  $I(Y;\hat{X},Z|T)-H(Z|T)$, $E[d(F(X,Y),\hat{X})]$, and the Markov relations.

First we bound the cardinality of the set $\mathcal{T}$.
Let us define the following $|\mathcal{X}|+5$ continuous functions of $p(u|t)$, $u \in \mathcal{U}$,
\begin{align*}
&f_{j}(p(u|t)) = \sum_{u \in \mathcal{U}}p(u|t)p(x|u,t),\ j=1,\ldots,|\mathcal{X}|-1, \\
&f_{|\mathcal{X}|}(p(u|t))  = H(X,Y|U,T=t)-H(Y|U,T=t)\\ & \qquad \qquad \qquad = H(X,U,Y|T=t)-H(U,Y|T=t),\\
&f_{|\mathcal{X}|+1}(p(u|t)) =H(\hat{X},Y|U,T=t)-H(Y|U,T=t)\\ & \qquad \qquad \qquad = H(\hat{X},U,Y|T=t)-H(U,Y|T=t),\\
&f_{|\mathcal{X}|+2}(p(u|t)) = H(Z,Y|U,T=t)-H(Y|U,T=t)\\ & \qquad \qquad \qquad = H(Z,U,Y|T=t)-H(U,Y|T=t),\\
&f_{|\mathcal{X}|+3}(p(u|t)) = I(Y;\hat{X},Z|T=t)-H(Z|T=t),\\
&f_{|\mathcal{X}|+4}(p(u|t)) = H(\hat{X}|U,Y,X,Z,T=t)\\ & \qquad \qquad \qquad =H(\hat{X},U,Y,X,Z|T=t)-H(U,Y,X,Z|T=t),\\
&f_{|\mathcal{X}|+5}(p(u|t)) = E[d(F(X,Y),\hat{X})|T=t].
\end{align*}
The corresponding averages are
\begin{align*}
& \sum_{t \in \mathcal{T}}p(t)f_{j}(p(u|t))=P_{X}(x),\ j=1,\ldots,|\mathcal{X}|-1, \\
& \sum_{t \in \mathcal{T}}p(t) f_{|\mathcal{X}|}(p(u|t))= H(X,U,Y|T)-H(U,Y|T), \\
& \sum_{t \in \mathcal{T}}p(t) f_{|\mathcal{X}|+1}(p(u|t))= H(\hat{X},U,Y|T)-H(U,Y|T), \\
& \sum_{t \in \mathcal{T}}p(t) f_{|\mathcal{X}|+2}(p(u|t))= H(Z,U,Y|T)-H(U,Y|T), \\
& \sum_{t \in \mathcal{T}}p(t) f_{|\mathcal{X}|+3}(p(u|t))= I(Y;\hat{X},Z|T)-H(Z|T), \\
& \sum_{t \in \mathcal{T}}p(t) f_{|\mathcal{X}|+4}(p(u|t))= H(\hat{X},U,Y,X,Z|T)-H(U,Y,X,Z|T), \\
& \sum_{t \in \mathcal{T}}p(t) f_{|\mathcal{X}|+5}(p(u|t))= E[d(F(X,Y),\hat{X})].
\end{align*}
According to the support lemma \cite{CsiszarBook}, we can deduce that there exist a new random variable $T'$ jointly distributed with $(X,Y,Z,U,\hat{X})$ whose alphabet size is $|\mathcal{T}'|= |\mathcal{X}|+5$, and numbers $\alpha_{i} \geq 0$ with $\sum_{i=1}^{|\mathcal{X}|+5}\alpha_{i} =1$ that satisfy
\begin{align*}
&\sum_{i=1}^{|\mathcal{X}|+5}\alpha_{i} f_{j}(P_{U|T'}(u|i)) = P_{X}(x),\ j=1,\ldots,|\mathcal{X}|-1, \\
&\sum_{i=1}^{|\mathcal{X}|+5}\alpha_{i}f_{|\mathcal{X}|}(P_{U|T'}(u|i)) = H(X,U,Y|T')-H(U,Y|T'),\\
&\sum_{i=1}^{|\mathcal{X}|+5}\alpha_{i}f_{|\mathcal{X}|+1}(P_{U|T'}(u|i)) = H(\hat{X},U,Y|T')-H(U,Y|T'),\\
&\sum_{i=1}^{|\mathcal{X}|+5}\alpha_{i}f_{|\mathcal{X}|+2}(P_{U|T'}(u|i)) = H(Z,U,Y|T')-H(U,Y|T'),\\
&\sum_{i=1}^{|\mathcal{X}|+5}\alpha_{i}f_{|\mathcal{X}|+3}(P_{U|T'}(u|i)) = I(Y;\hat{X},Z|T')-H(Z|T'),\\
&\sum_{i=1}^{|\mathcal{X}|+5}\alpha_{i}f_{|\mathcal{X}|+4}(P_{U|T'}(u|i)) = H(\hat{X},U,Y,X,Z|T')-H(U,Y,X,Z|T'),\\
&\sum_{i=1}^{|\mathcal{X}|+5}\alpha_{i}f_{|\mathcal{X}|+5}(P_{U|T'}(u|i)) = E[d(F(X,Y),\hat{X})].
\end{align*}
Note that
\begin{align*}
&H(X,U,Y|T')-H(U,Y|T') \\
&= H(X,U,Y|T)-H(U,Y|T)\\
& \overset{(\star)}{=} H(X,Y|U)-H(Y|U),
\end{align*}
where $(\star)$ follows from the Markov chain $T-U-X-(Y,Z)$.
Similarly, from the Markov chains $T-U-X-(Y,Z)$ and $\hat{X}-(U,Y)-(X,Z,T)$, we have that $H(\hat{X},U,Y|T')-H(U,Y|T')=H(\hat{X},Y|U)-H(Y|U)$ and $H(Z,U,Y|T')-H(U,Y|T')=H(Z,Y|U)-H(Y|U)$.
Since $P_{X}(x)$ is preserved, $P_{X,Y}(x,y)$ is also preserved.  Thus, $H(X|Y)$ is preserved.

Next we bound the cardinality of the set $\mathcal{U}$.
For each $t' \in \mathcal{T}'$, we define the following $|\mathcal{X}|+4$ continuous functions of $p(x|t',u)$,\ $x \in \mathcal{X} $,
\begin{align*}
&f_{j}(p(x|t',u)) = p(x|t',u),\ j=1,\ldots,|\mathcal{X}|-1, \\
&f_{|\mathcal{X}|}(p(x|t',u)) = H(X,Y|T'=t',U=u)-H(Y|T'=t',U=u),\\
&f_{|\mathcal{X}|+1}(p(x|t',u)) = H(\hat{X},Y|T'=t',U=u)-H(Y|T'=t',U=u),\\
&f_{|\mathcal{X}|+2}(p(x|t',u)) = H(Z,Y|T'=t',U=u)-H(Y|T'=t',U=u),\\
&f_{|\mathcal{X}|+3}(p(x|t',u)) = H(\hat{X},Y,X,Z|T'=t',U=u)-H(Y,X,Z|T'=t',U=u),\\
&f_{|\mathcal{X}|+4}(p(x|t',u)) = E[d(F(X,Y),\hat{X})|T'=t',U=u].
\end{align*}
Similarly to the previous part in bounding $|\mathcal{T}|$, there exists a new random variable $U'|\{T'=t'\} \sim p(u'|t')$ such that $|\mathcal{U}'| =  |\mathcal{X}|+4$ and $p(x|t')$, $H(X,Y|T'=t',U)-H(Y|T'=t',U)$, $H(\hat{X},Y|T'=t',U)-H(Y|T'=t',U)$, $H(Z,Y|T'=t',U)-H(Y|T'=t',U)$, $H(\hat{X},Y,X,Z|T'=t',U)-H(Y,X,Z|T'=t',U)$, and $E[d(F(X,Y),\hat{X})|T'=t']$ are preserved.

By setting $U'' =(U',T')$ where $\mathcal{U}'' = \mathcal{U}' \times \mathcal{T}'$, we have that $T'-U''-X-(Y,Z)$ forms a Markov chain. To see that the Markov chain $\hat{X}-(U'',Y)-(X,Z,T')$ also holds, we consider
\begin{align*}
&I(\hat{X};X,Z,T'|U'',Y) \\
& = I(\hat{X};X,Z|U',T',Y) \\
&= H(\hat{X}|U',T',Y) - H(\hat{X}|U',T',Y,X,Z)\\
&\overset{(a)}{=} H(\hat{X}|U,T',Y) - H(\hat{X}|U,T',Y,X,Z)\\
&\overset{(b)}{=} H(\hat{X}|U,T,Y) - H(\hat{X}|U,T,Y,X,Z)\\
&\overset{(c)}{=} 0,
\end{align*}
where $(a)$ follows from preservation by $U'$, $(b)$ follows from preservation by $T'$, and $(c)$ from the Markov chain $\hat{X}-(U,Y)-(X,Z,T)$.

Furthermore, we have the following preservations by $U''$,
\begin{align*}
&H(X,Y|U'')-H(Y|U'') \\
& = H(X,Y|U',T')-H(Y|U',T')\\
& \overset{(a)}{=} H(X,Y|U,T')-H(Y|U,T')\\
& \overset{(b)}{=} H(X,Y|U,T)-H(Y|U,T)\\
& \overset{(c)}{=} H(X,Y|U)-H(Y|U),
\end{align*}
where $(a)$ follows from preservation by $U'$, $(b)$ follows from preservation by $T'$, and $(c)$ follows from the Markov chain $T-U-X-(Y,Z)$.
Similarly, from preservation by $U'$ and $T'$, and the Markov chain $T-U-X-(Y,Z)$ and $\hat{X}-(U,Y)-(X,Z,T)$, we have that $H(\hat{X},Y|U'')-H(Y|U'') =H(\hat{X},Y|U)-H(Y|U)$ and $H(Z,Y|U'')-H(Y|U'')=H(Z,Y|U)-H(Y|U)$.

 Therefore, we have shown that $T \in \mathcal{T}$ and $U \in \mathcal{U}$ may be replaced by $T' \in \mathcal{T}'$ and $U'' \in \mathcal{U}''$ satisfying
 \begin{align*}
 |\mathcal{T}'|&= |\mathcal{X}|+5, \\
 |\mathcal{U}''| &= |\mathcal{T}'||\mathcal{U}'|= (|\mathcal{X}|+5)(|\mathcal{X}|+4),
 \end{align*}
and  preserving the terms $H(X,Y|U)-H(Y|U)$, $H(\hat{X},Y|U)-H(Y|U)$, $H(Z,Y|U)-H(Y|U)$,  $I(Y;\hat{X},Z|T)-H(Z|T)$, $E[d(F(X,Y),\hat{X})]$, and the Markov relations.


\begin{thebibliography}{1}
\bibitem{WynerZiv} A. D. Wyner and J. Ziv, ``The rate distortion function for source coding with side information at the decoder," \emph{IEEE Trans. Inf. Theory}, vol. IT-22, pp. 1-10, Jan. 1976.
    \bibitem{Yamamoto} H. Yamamoto, ``Wyner-Ziv theory for a general function of the correlated sources," \emph{IEEE Trans. Inf. Theory}, vol.28, no.5, pp.803-807, Sep 1982.
\bibitem{Merhav_a}  N. Merhav, ``On joint coding for watermarking and encription," \emph{IEEE Trans. Inf. Theory}, vol. 52, pp. 190-205, Jan. 2006.
\bibitem{Merhav_b}  N. Merhav, ``On the shannon cipher system with a capacity-limited key-distribution channel," \emph{IEEE Trans. Inf. Theory}, vol. 52, pp. 1269-1273, Mar. 2006.
\bibitem{Schieler} C. Schieler and P. Cuff, ``Rate-distortion theory for secrecy systems," \emph{http://arxiv.org/abs/1305.3905}, 2013.
\bibitem{EkremUlukus} E. Ekrem and S.Ulukus,  ``Secure lossy source coding with side information," in \emph{Proc. Allerton Conf. Commun. Control Comput.}, 2011.
    \bibitem{TandonSankarPoor} R. Tandon, L. Sankar, and H. V. Poor, ``Discriminatory lossy source coding: Side information privacy," \emph{IEEE Trans. Inf. Theory}, vol. 59, pp. 5665-5677, Sep. 2013.
\bibitem{HeegardBerger} C. Heegard and T. Berger, ``Rate distortion when side information may be absent," \emph{IEEE Trans. Inf. Theory}, vol. 31, no. 6, pp. 727–734, Nov. 1985
\bibitem{PrabhakaranRamchandran} V. Prabhakaran and K. Ramchandran, ``On secure distributed source coding," in \emph{Proc. IEEE Inf. Theory Workshop}, 2007, pp. 442-447.
\bibitem{GunduzErkipPoor} D. G\"{u}nd\"{u}z, E. Erkip and H. V. Poor, ``Lossless compression with security constraints," in \emph{Proc. IEEE ISIT}, 2008, Toronto, pp. 111-115.
\bibitem{TandonUlukusRamchandran} R. Tandon, S. Ulukus and K. Ramchandran, ``Secure source coding with a helper," \emph{IEEE Trans. Inf. Theory}, vol. 59, no. 4, pp. 2178-2187, 2013.
\bibitem{VillardPiantanida} J. Villard and P. Piantanida, ``Secure multiterminal source coding with side information at the eavesdropper," \emph{IEEE Trans. Inf. Theory}, vol. 59, no. 6, June 2013.
\bibitem{ElGamalKim} A. El Gamal and Y.-H. Kim, \emph{Network Information Theory}, Cambridge University Press, 2011.
\bibitem{Steinberg} Y. Steinberg, ``Coding and common reconstruction," \emph{IEEE Trans. Inf. Theory}, vol. 55, no. 11, 2009.
\bibitem{Lapidoth} A. Lapidoth, A. Mal{\"a}r, and M. Wigger, ``Constrained Wyner-Ziv coding," in \emph{Proc. IEEE ISIT}, 2011, St.Petersburg, Russia.
\bibitem{WeissmanElGamal} T. Weissman and A. El Gamal, ``Source coding with limited-look-ahead side information at the decoder," \emph{IEEE Trans. Inf. Theory}, vol. 52, no. 12, pp. 5218-–5239, 2006.
\bibitem{CsiszarBook} I. Csisz\'{a}r and J. K\"{o}rner. \emph{Information Theory: Coding Theorems for Discrete Memoryless Systems}. Cambridge University Press, 2011.
\end{thebibliography}
\end{document}